\documentclass[letterpaper, 10 pt, conference]{ieeeconf}  

\IEEEoverridecommandlockouts                              
\usepackage{color}
\usepackage{fixltx2e}
\usepackage{graphics} 
\usepackage{epsfig} 
\usepackage{amsmath} 
\usepackage{amssymb}  
\usepackage[colorinlistoftodos,prependcaption,textsize=tiny]{todonotes}
\usepackage{regexpatch}\usepackage[norelsize,ruled,vlined]{algorithm2e}
\usepackage{graphicx}
\PassOptionsToPackage{usenames,dvipsnames,svgnames}{xcolor}  
\usepackage{tikz}
\usetikzlibrary{arrows,positioning,automata}

\renewcommand{\baselinestretch}{0.99}
\newtheorem{theorem}{Theorem}
\newtheorem{remark}{Remark}
\newtheorem{lemma}{Lemma}

\newcommand{\mb}{\mathbf}
\newcommand{\mc}{\mathcal}

\presetkeys%
    {todonotes}%
    {inline,backgroundcolor=orange}{}

\title{On the Limited Communication Analysis and Design \\ for Decentralized Estimation\vspace{-0.3cm}}
\author{Andreea B. Alexandru $^{\dag}$ $\quad$ S\'ergio Pequito $^{\dag}$ $\quad$ Ali Jadbabaie  $^{\ddagger}$ $\quad$ George J. Pappas $^{\dag}$
 \thanks{This work was supported in part by the TerraSwarm Research Center, one of six centers supported by the STARnet phase of the Focus Center Research Program (FCRP) a Semiconductor Research Corporation program sponsored by MARCO and DARPA.}
\thanks{
$^{\dag}$Department of Electrical and Systems Engineering, School of Engineering and Applied Science,
University of Pennsylvania
\newline \indent
$^{\ddagger}$Department of Civil and Environmental Engineering, and Institute for Data, Systems, and Society, Massachusetts Institute of Technology}\vspace{-0.4cm}
}

\date{} 

\begin{document}

\maketitle

\thispagestyle{empty}
\pagestyle{empty}

\begin{abstract}
This paper pertains to the analysis and design of decentralized estimation schemes that make use of limited communication. Briefly, these schemes equip the sensors with scalar states that iteratively merge the measurements and the state of other sensors to be used for state estimation. Contrarily to commonly used distributed estimation schemes, the only information being exchanged are scalars, there is only one common time-scale for communication and estimation, and the retrieval of the state of the system and sensors is achieved in finite-time. We extend previous work to a more general setup and provide necessary and sufficient conditions required for the communication between the sensors that enable the use of limited communication decentralized estimation~schemes. Additionally, we discuss the cases where the sensors are memoryless, and where the sensors might not have the capacity to discern the contributions of other sensors. Based on these conditions and the fact that communication channels incur a cost, we cast the problem of finding the minimum cost communication graph that enables limited communication decentralized estimation schemes as an integer programming problem.
\end{abstract}

\vspace{-0.1cm}
\section{Introduction}

Sensors are often geographically deployed to collect measurements over large-scale networked dynamical systems, which are used by state estimators (implementing state observers) to retrieve an estimate of the overall state of the system. Then, the estimate is provided to the actuator that implements a controller to steer the dynamical system to the desired state~\cite{guptaPhd,GuptaStateEstPowerSys,karPhd,khanPhd}. Estimators can be full-state, reduced and extended state observers that implicitly explore the trade-offs between communication and estimation~\cite{wang2017distributed,mitra2016distributed}. 
In the context of distributed estimation, additional information used by the state estimators is shared to improve the quality of the estimate. For example, they can share either the estimate, the error between predicted state observation and the measurement, or the innovation used as part of the state estimation process~\cite{Garin2010,DasMoura,khanPhd,Subbotin}. 
The information is shared by resorting to communication between different sensors' computational units, and the communication capabilities impact the ability to retrieve the estimate of the state. Therefore, it is fundamental to understand which communication is required to ensure a successful recovery of the system state~\cite{KhanJadbabaie,Khan10}. 

Most of the observers implemented in large-scale systems require a large amount of information being exchanged through communication (i.e., the state, the error between measurement and predicted state observation). 
Furthermore, the estimators are shown to be asymptotically stable (not always with an arbitrary error decay) which might restrict the actuation performance in the context of large-scale networked dynamical systems. To overcome such limitations, in~\cite{AlexandruCDC16}, we proposed an approach that equips the sensors with scalar states which are exchanged with other sensors, and together with sensors measurements, suffice to retrieve in finite-time the state of the networked dynamical system and those of the sensors, which we refer to as \emph{limited communication decentralized estimation scheme}. Subsequently, the information being exchanged between sensors is reduced to the bare minimum, and the communication topologies analyzed are designed to ensure sensor-network state recovery.

In this paper, we seek to better understand the restrictions and trade-offs of the limited communication decentralized estimation scheme. Specifically, the main contributions of this paper are: (\emph{i}) we waive some implicit assumptions made in~\cite{AlexandruCDC16} about the communication scheme performed by the sensors (which are in general only sufficient, as we emphasize in Remark~\ref{re:particularization}); (\emph{ii}) we explore the implications in two different setups: (\emph{a}) the sensors are \emph{memoryless} (i.e., they do not keep track of their previous state); and (\emph{b}) sensors might not have the capacity to discern the contributions and/or state of other sensors (e.g., those relying on radio technology); and (\emph{iii})~we leverage these new conditions to cast the problem of determining the minimum communication cost required to deploy a limited communication decentralized estimation scheme as an integer programming problem.

%
%

\vspace{-0.1cm}
\section{Problem Statement}

Let the evolution of a (possibly) large-scale networked dynamical system be captured by
\begin{equation}
\mb x[k+1]=\mb A\mb x[k], \quad k=0,1,\ldots
\label{dynamics}
\end{equation}
where $\mb x[k]\in\mathbb{R}^{n\times 1}$ is the state of the system. Consider $m$~sensors with measurements $y_i\in\mathbb{R}$ described as follows:
\begin{equation}
y_i[k]=\mb c_i^\intercal \mb x[k], \quad i=1,\ldots,m,
\label{output}
\end{equation}
where $\mb c_i\in\mathbb{R}^{n\times 1}$ is the output vector describing the contributions of the different observed state variables. We assume that $(\mb A,\mb C=[\mb c_1^\intercal\ \mb c_2^\intercal \ \ldots \ \mb c_m^\intercal]^\intercal)$ is observable, but not necessarily observable from a specific sensor $i$, i.e., $(\mb A, \mb c_i^\intercal)$ is not necessarily observable. 

In the limited communication decentralized estimation scheme, we consider that the sensors possess a scalar state and can communicate with each other. During this process, they share their states, which enables the retrieval of the state of both the networked dynamical system and the sensors. The communication capabilities are captured by a directed \emph{communication graph} $\mc G=(\mc V,\mc E)$, where the set of vertices $\mc V=\{1,\ldots,m\}$ labels the $m$ sensors, and an edge $(i,j)\in\mc E$ translates in the capability of sensor $j$ to receive data from sensor $i$. Besides, each sensor $i$ computes a linear combination of the (scalar) measurement and the scalar data $z_j\in \mathbb{R}$ received from the neighboring sensors, i.e., $j\in \mathbb{N}_i^-$, which can be described as follows:\vspace{-0.1cm}
\begin{equation}
z_i[k+1]=y_i[k]+\sum_{j\in\mathbb N_i^-}w_{ij}z_j[k], \ i\in\mc{V},
\label{dynSensors}
\end{equation}\vspace{-0.4cm}

\noindent where $\mathbb N_i^-=\{j\in\mc{V}: (i,j)\in \mc E\}$ are the indices of the \mbox{\mbox{in-neighbors}} of sensor $i$ given by the communication graph~$\mc G$.
Subsequently, we can write~\eqref{dynamics}-\eqref{dynSensors} using the following compact representation:\vspace{-0.1cm}
\begin{equation}
\tilde{\mb x}[k+1]=\left[\begin{array}{cc}\mb A & \mb{0}_{n\times m}\\ \mb C & \mb W(\mc G) \end{array}\right]\tilde{\mb x}[k] =: \tilde{\mb A}(\mc G) \tilde{\mb x}[k],
\label{dynAugmented}
\end{equation}\vspace{-0.4cm}

\noindent where $\tilde{\mb x}=[x_1\ \ldots\ x_n \ z_1 \ldots z_m]^\intercal$ is the augmented system's state and $\mb W(\mc G)$ the dynamics between sensors induced by the communication graph, i.e., $[\mb W(\mc G)]_{ij}=w_{ij}$ when $(i,j)\in \mc E$ and zero otherwise. It is worth noticing that some of the weights $w_{ij}$ may be set to zero, and, in particular, if $w_{ii}=0$ for $i\in\{1,\ldots,m\}$ then we are dealing with \emph{memoryless} sensors that work as relays -- which cannot be addressed by the setup explored in~\cite{AlexandruCDC16}.
Additionally, the augmented system's output is described as follows:\vspace{-0.1cm}
\begin{equation}
\tilde y_i[k]=\left[\begin{array}{cc} -\ \mb c_i^\intercal \ - & \mb{0}_{1\times m}\\
\mb{0}_{|\mc J_i|\times n} &\mb{I}_{m}^{\mc J_i}\end{array}\right] \tilde{\mb x}[k] =: \tilde{\mb C}_i \tilde{\mb x}[k], \ i\in\mc V,
\label{outputAugmented}
\end{equation}\vspace{-0.4cm}

\noindent where $\mb{I}_{m}^{\mc J_i}$ is the sub-matrix containing the rows of the $m\times m$ identity matrix with indices in $\mc J_i\subset\{1,\ldots,m\}$. In particular, $\mathcal J_i=\mathbb{N}_i^-$ when the linear combination of incoming sensor's states is performed (locally) at sensor $i$, or $\mathcal J_i=\{i\}$ when sensors do not have the capacity to discern the contributions and/or state of other sensors (e.g., those relying on radio technology). The latter case cannot be addressed by the setup explored in~\cite{AlexandruCDC16}.

In this paper, we seek solutions to the following problems:

\emph{Problem~1:} Characterize the necessary and sufficient conditions that must be satisfied by $\mc G$ (and, subsequently, by $\mb W(\mc G)$) ensuring that $({\tilde{\mathbf{A}}}(\mc G),{\tilde{\mathbf{C}}}_i)$ is observable.

In particular, we provide the characterization required in the memoryless sensor scenario, and in the case where a sensor only has access to its own state. 
Next, we propose to determine communication topologies that ensure the necessary and sufficient conditions required to solve the previous problem, while minimizing the communication cost between the different sensors:

\emph{Problem~2:} Let $\mb \Omega_{e}\in \mathbb{R}^+_{\geq0}\cup\{\infty\}$ be the communication cost incurred by establishing a communication link $e = (i,j)\in \mathcal E$ between the sensors $i,j\in\mathcal V$ to obtain a communication graph $\mathcal G=(\mathcal V,\mathcal E)$. We aim to determine $\mathcal E$ that solves the following optimization problem:\vspace{-0.1cm}
\[\min\limits_{\mathcal E}~ \sum_{e\in \mathcal E} \mb \Omega_{e}~
\text{s.t.}~\mathcal G \text{ satisfies conditions from \emph{Problem~1}.}
\]

\section{Terminology and Previous Results}\label{sec:termPR} 

In what follows, we rely on \emph{structural systems} theory~\cite{dionSurvey} to assess system theoretical properties by considering only the inter-dependencies between states and sensors.
One such system property is that of \emph{structural observability} that considers the sparsity binary patterns $(\bar{\mb A},\bar{\mb C})$, where an entry in these matrices is zero if there is no direct dependency between two (state or sensor) variables and one otherwise~\cite{Lin_1974,Shields_Pearson:1976}. A pair $(\bar{\mb A},\bar{\mb C})$ is structurally observable if there exists an observable pair $(\mb A,\mb C)$ such that the zero entries in $(\bar{\mb A},\bar{\mb C})$ are also zero in $(\mb A,\mb C)$. Subsequently, it can be proved that if such an observable pair exists, then almost all possible pairs satisfying the sparsity pattern are also observable. Furthermore, structural properties (e.g., structural observability) are necessary to ensure non-structural properties (observability). Therefore, in Section~\ref{sec:main}, we rely on structural systems to ensure first the necessary conditions, and then we show that in fact these are also sufficient.

One of the key features of structural systems theory is that we can interpret the sparsity patterns $(\bar{\mb A},\bar{\mb C})$ as a directed state graph $\mathcal D(\bar{\mb A})\equiv(\mathcal X,\mathcal E_{\mathcal X,\mathcal X})$ and \emph{\mbox{state-output} graph} $\mathcal D(\bar{\mb A},\bar{\mb C}) \equiv(\mathcal X\cup\mathcal Y,\mathcal E_{\mathcal X,\mathcal X}\cup \mathcal E_{\mathcal X,\mathcal Y})$, where the vertices are labeled by the states and sensors and the edges capture the inter-dependencies between state and sensor variables as follows: $\mathcal E_{\mathcal X,\mathcal X}=\{(x_i,x_j): [\bar{\mb A}]_{ji}\neq 0\}$ and $\mathcal E_{\mathcal X,\mathcal Y}=\{(x_i,y_j): [\bar{\mb C}]_{ji}\neq 0\}$. We will use $\mc E_{\mc X,\mc X\cup \mc Y}:=\mathcal E_{\mathcal X,\mathcal X}\cup \mathcal E_{\mathcal X,\mathcal Y}$ for brevity.
Additionally, we can use graph-theoretical notions, e.g., paths and cycles, to address the structural properties. In particular, to characterize structural observability, we introduce the notion of bipartite graph associated with the state graph and \mbox{state-output} graph.~
The \emph{state bipartite graph} $\mathcal B(\bar{\mb A})\equiv\mc{B}(\mc{X},\mc{X},\mc{E_{X,X}})$ (resp., the \emph{\mbox{state-output} bipartite graph} $\mathcal B(\bar{\mb A},\bar{\mb C})\equiv\mc{B}(\mc{X},\mc{X}\cup \mc{Y},\mc{E_{X,X\cup Y}})$), consists of two sets $\mc{X}$ (resp., $\mc{X}$ and $\mc{X}\cup \mc{Y}$) that can be graphically interpreted and to which we refer to as left and right set of vertices. Edges between the left and right set of vertices encode the dependencies described by the edge-set of the directed state graph (resp., directed \mbox{state-output} graph). 
Also, due to the correspondence between these edges, paths and cycles in the \mbox{state-output} graph can be captured by subsets of \mbox{vertex-disjoint} edges in the state and \mbox{state-output} bipartite graph, which are referred to as \emph{matchings}, and the subset with the largest number of edges referred to as \emph{maximum matching}. Consequently, those left (resp., right) vertices in the state and \mbox{state-output} bipartite graph that do not belong to any edge in the matching are referred to as \emph{\mbox{left-unmatched}} (resp., \emph{\mbox{right-unmatched}}) vertices. Accordingly, we have the following result:

\begin{lemma}[\cite{PequitoJournal}]\label{maxMatPathCycle}
Consider the digraph $\mc D(\bar{\mb A},\bar{\mb C})\equiv(\mc X\cup\mc Y,\mc E_{\mc X,\mc X\cup\mc Y})$ and let $M^{\ast}$ be a maximum matching associated to the \mbox{state-output} bipartite graph $\mc B(\bar{\mb A},\bar{\mb C})\equiv \mc B(\mc X,\mc X\cup\mc Y,\mc E_{\mc X,\mc X\cup\mc Y})$. Then, the digraph $\mc D\equiv(\mc X\cup\mc Y, M^{\ast})$ comprises a disjoint union of cycles and elementary paths, from the left-unmatched vertices to the \mbox{right-unmatched} vertices of $M^{\ast}$, that span $\mc D(\bar{\mb A},\bar{\mb C})$. Moreover, such a decomposition is \emph{minimal}, in the sense that no other spanning subgraph decomposition of $\mc D(\bar{\mb A},\bar{\mb C})$ into elementary paths and cycles contains strictly fewer elementary paths.
\hfill $\diamond$
\end{lemma}

The different graph-theoretic concepts can come together to assess structural observability of $(\bar{\mb A},\bar{\mb C})$ as follows.

\begin{theorem}[\cite{AlexandruCDC16}]\label{thm:structObs}
Let $\mc{D}(\bar{\mb A},\bar{\mb C}) \equiv (\mc{X}\cup\mc{Y},\mc E_{\mc X,\mc X\cup\mc Y})$ denote the \mbox{state-output} digraph and $\mc{B}(\bar{\mb A}, \bar{\mb C})$ the \mbox{state-output} bipartite representation. The pair $(\bar{\mb A},\bar{\mb C})$ is structurally observable if and only if the following two conditions hold:
\begin{enumerate}
\item[(i)] there is a path from every state vertex to an output vertex in $\mc D(\bar{\mb A}, \bar{\mb C})$; and
\item[(ii)] there exists a maximum matching $M^{\ast}$ associated to $\mc{B}(\bar{\mb A}, \bar{\mb C})$ such that the \mbox{left-unmatched} vertices $\mc U_L (M^{\ast}) = \emptyset$. \hfill $\diamond$
\end{enumerate}
\end{theorem}

Therefore, as previously mentioned, we can build upon these results to analyze and design the limited communication decentralized estimation scheme. In~\cite{AlexandruCDC16}, we have provided necessary and sufficient conditions that $\mathcal G$ needs to satisfy to ensure observability of $(\tilde{\mb A}(\mc G), \tilde{\mb C}_i)$ under the following simplifying implicit assumption.

\textit{Implicit Assumption~\cite{AlexandruCDC16}:} Each sensor retains its previous state that is always weighted in the sensor dynamics~\eqref{dynSensors}, i.e., $w_{ii}\neq 0$, which implies that in~\eqref{dynAugmented} we have $[\text{diag}(\mb W(\mc G))]_{ii}\neq 0$ for all $i=1,\ldots, m$. \hfill $\circ$

In other words,~\cite{AlexandruCDC16} excludes the case of memoryless sensors, which we address in this paper. We also explore other setups, 
e.g., when the sensors are not able to differentiate the individual contributions of other sensors due to the technology used. Next, we state two of the main results in~\cite{AlexandruCDC16} for ease of comparison with the main results attained~in this work, where we waive the implicit assumption stated above.

\begin{theorem}[\cite{AlexandruCDC16}]\label{thm:decstructobs}
Let $\mc D(\tilde{\mb A}(\mc G))\equiv(\mc V\equiv(\mc X\cup \mc Z),\mc E_{\mc V,\mc V})$ be the state digraph, where $\mc X$ corresponds to the labels of the state vertices and $\mc Z$ to the labels of the sensors' states. In addition, let $\mc N_{i}^{-}=\{v\in\mc V: (v,z_i)\in\mc E\}$ be the set of \mbox{in-neighbors} of a vertex $z_i$ representing a sensor in $\mc D(\tilde{\mb A}(\mc G))$, $i = 1,\ldots,m$. The following two conditions are necessary and sufficient to ensure that~$(\tilde{\mb A}(\mc G),\tilde{\mb C}_i)$, for $i=1,\ldots,m$, is generically observable:
\begin{enumerate}
\item[(i)] for every $z\in \mc Z$ there must exist a directed path from any $v\in\mc V$;
\item[(ii)] for every $z\in \mc Z$ there must exist a set of \mbox{left-unmatched} vertices $\mc U_L$, associated with a maximum matching of the bipartite representation of $\mc D(\tilde{\mb A}(\mc G))$, such that $\mc U_L\subset \mc N_i^-$ and $\mc U_L\cap \mc X = \emptyset$.~\hfill~$\diamond$
\end{enumerate}
\end{theorem}

Hence, Theorem~\ref{thm:decstructobs} can be used to obtain the next result to \emph{Problem~1} under the implicit assumption stated above.

\begin{theorem}[\cite{AlexandruCDC16}]\label{thm:almostAll}
If $(\mb A,\mb C)$ is observable and $(\tilde{\mb A}(\mc G),\tilde{\mb C}_i)$ is structurally observable $\forall i=1,\ldots,m$, then almost all realizations of $\mb W(\mc G)$ ensure that $(\tilde{\mb A}(\mc G),\tilde{\mb C}_i)$ is observable.
\hfill $\diamond$
\label{sufficiency}
\end{theorem}

%
For brevity's sake, we will use the shortened notation $\tilde{\mb A} = \tilde{\mb A}(\mc G)$ in the rest of the paper.

\vspace{-0.1cm}
\section{Limited Communication Analysis and Design}

In this section, we introduce the main results of this paper. Lemma~\ref{le:iffstructObs} shows that the sensor capabilities impose strong constraints on the network's structure required to ensure structural observability. This technical result plays a key role in understanding Theorem~\ref{thm:iffW}, which states the necessary and sufficient conditions required to address \emph{Problem~1}. Specifically, it provides the conditions for the communication graph $\mathcal G$ such that $(\tilde{\mb A},\tilde{\mb C}_i)$ is observable, $i=1,\ldots,m$.~Next, we consider the design of communication graphs that attain the former conditions, while minimizing the total cost incurred by the communication between the sensors (\emph{Problem~2}). In particular, we cast the problem as an integer programming problem that can be solved with \mbox{off-the-shelf} solvers.

We start by showing that the structural observability of a pair $(\bar{\mb{A}},\bar{\mb{C}})$, that is often assessed through the \mbox{state-output} graph properties (as captured in Theorem~\ref{thm:structObs}), can enforce a particular structure of $\bar{\mb A}$ under certain sensing capabilities.

\begin{lemma}\label{le:iffstructObs} Let $\mb e_i\in \mathbb{R}^{p\times 1}$ be the canonical column-vector with one in the $i$'th position and the remaining entries equal to zero.
Given a structured adjacency matrix of a graph, $\bar{\mb M}\in\{0,1\}^{p \times p}$, the pairs $(\bar {\mb M},\mb e_i^\intercal)$, for $i=1,\ldots,p$, are structurally observable if and only if the associated state digraph $\mc D(\bar{\mb M})$ is strongly connected and spanned by a disjoint union of cycles. \hfill $\diamond$
\end{lemma}

\begin{proof}
\emph{(Necessity)}~Assume $(\bar {\mb M},\mb e_i^\intercal)$ is structurally observable, for all $ i=1,\ldots,p$. Suppose by contradiction that $\mc D(\bar{\mb M})$ is not strongly connected. If $\mc D(\bar{\mb M})$ is not strongly connected, then the \mbox{state-output} digraph $\mc D(\bar{\mb M}, \mb e_i^\intercal)$ is also not strongly connected and its directed acyclic representation contains a number of strongly connected components. Since the output vertex in $\mc D(\bar{\mb M}, \mb e_i^\intercal)$ is one vertex $y_i$ connected only to vertex~$x_i$, then it readily follows that condition (i) of Theorem~\ref{thm:structObs} cannot hold.

Next, we prove that $\mc D(\bar{\mb M})$ is spanned by a disjoint union of cycles.
Consider condition (ii) of Theorem~\ref{thm:structObs}: there exists a maximum matching in $\mc B(\bar{\mb M}, \mb e_i^\intercal)$ such that there is no \mbox{left-unmatched} vertex in $\mc X$. There are two possibilities for the maximum matchings in $\mc B(\bar{\mb M})$: (a) $\mc U_L = \emptyset$ and (b) $\mc U_L \neq \emptyset$. Case (a) means that there is a perfect matching in $\mc B(\bar{\mb M})$, i.e., from Lemma~\ref{maxMatPathCycle}, $\mc D(\bar{\mb M})$ is spanned by a disjoint union of cycles.

We want to prove now that case (b) cannot happen when $(\bar {\mb M},\mb e_i^\intercal)$ is structurally observable. Let us first address the case when $|\mc {U}_L| = 1$. Let $\mc X$ represent the set of vertices in the graph described by $\bar{\mb M}$. Consider a maximum matching $M^1$ in $\mc B(\bar{\mb M})$ that has $\mc U_L(M^1) = \{x_j\}$ and $\mc U_R(M^1) = \{x_k\}$, for some vertices $x_j,x_k\in\mathcal X$. Since $\mc D(\bar{\mb M})$ is strongly connected, then, there exists a neighbor $x_l$ of $x_k$ such that $(x_l,x_k) \in \mc E_{\mc X,\mc X}$. Let $M^2$ be another maximum matching in $\mc B(\bar{\mb M})$ such that $\mc U_L = \{x_l\}$ and $\mc U_R = \{x_m\}$, for some vertex $x_m\in\mathcal X$. By Lemma~4 in \cite{PequitoJournal}, there exists a maximum matching $M^\Delta$ such that $\mathcal U_L(M^\Delta) = \{x_l\}$ and $\mathcal U_R(M^\Delta) = \{x_k\}$. However, $M^\ast = M^\Delta \cup \{(x_l,x_k)\}$ is also a maximum matching, in fact, a perfect matching, which leads to a contradiction of the fact that $M^\Delta$ is a maximum matching. Therefore, the set of \mbox{left-unmatched} vertices has to be empty, meaning a maximum matching is also a perfect matching, leading to the fact that $D(\bar{\mb M})$ is spanned by a disjoint union of cycles. 
Now, for the case when $|\mc {U}_L| > 1$, we can iteratively find augmented paths \cite {DiestelGT05} and construct larger cardinality maximum matchings, while thus reducing the cardinality of the set of \mbox{left-unmatched} vertices with respect to those matchings, until $|\mc {U}_L| = 1$.

\emph{(Sufficiency)}~Assume $\mc D(\bar{\mb M})$ is strongly connected and spanned by a disjoint union of cycles. It follows that also $\mc D(\bar{\mb M},\mb e_i^\intercal)$ is strongly connected and condition (i) from Theorem~\ref{thm:structObs} is satisfied. Since $\mc D(\bar{\mb M})$ is spanned by a disjoint union of cycles, by Lemma~\ref{maxMatPathCycle}, there exists a perfect matching $M^\ast$ (which is also a maximum matching) in $\mc B(\bar{\mb M},\mb e_i^\intercal)$. This implies condition (ii), i.e., $\mc U_L(M^\ast) = \emptyset$.
\end{proof}

\begin{remark}
In the proof of Lemma~\ref{le:iffstructObs}, the technical challenge is to show that the state vertices in the \mbox{state-output} bipartite graph cannot be always matched by an edge whose right-vertex is a sensor vertex (when the state graph is strongly connected), which implies that those states need to be always matched by edges whose end-points are state vertices. Therefore, by leveraging Lemma~\ref{maxMatPathCycle}, it follows that the state graph has to be spanned by cycles. \hfill $\diamond$
\end{remark}

\vspace{-0.1cm}
\section{Main Results}\label{sec:main}

Before we present the solution to the former problem, we need to review the notion of \emph{linking} (see, for instance,~\cite{dionSurvey}) from the vertices in the state digraph to the vertices in the communication digraph in the \mbox{state-output} graph associated with the augmented system. Specifically, a linking~$\mc P$ is a set of \mbox{vertex-disjoint} and simple paths in $\mc D(\tilde{\mb A})$, from the vertices in $\mc D(\bar{\mb A})$ to the vertices in $\mc D(\mb W(\mc G))$. Additionally, for each sensor $i$, we denote by $\mc P_i$ the \mbox{communication-linking} (a linking where both the starting and ending vertices in the \mbox{vertex-disjoint} simple paths belong to the communication graph) from the set of sensor vertices that belong to the edges in a maximum matching of $\mc B(\bar{\mb A}, \bar{\mb C})$ to a subset of \mbox{in-neighbors} of sensor $i$ ($\mathcal J_i$) with equal cardinality. In particular, there are as many of those sensor vertices as \mbox{left-unmatched} vertices in a maximum matching associated to $\mc B(\bar{\mb A})$ due to the observability of $(\mb A,\mb C)$, and its structural observability, as prescribed by Theorem~\ref{thm:structObs}. 
Consequently, the solution to \emph{Problem~1} can be formally stated as follows.

\begin{theorem}\label{thm:iffW}
Consider the system~\eqref{dynAugmented}-\eqref{outputAugmented}. For $(\mb A,\mb C)$ observable, the pair $(\tilde{\mb A}, \tilde{\mb C}_i)$ is observable for all $ i=1,\ldots,m$, if and only if $\mc D(\mb W(\mc G))$ is strongly connected and there exists a linking $\mathcal P_i$ such that $\mc D(\mb W(\mc G\setminus \mathcal P_i))$ is spanned by a disjoint union of cycles, for all $i=1,\ldots,m$. \hfill $\diamond$
\end{theorem}

\begin{proof}
\emph{(Necessity)}~It is enough to show that one condition from Theorem~\ref{thm:decstructobs} or Theorem~\ref{thm:almostAll} does not hold when $\mc D(\mb W(\mc G\setminus \mc P_i))$ is not spanned by a disjoint union of cycles or that $\mathcal D(\mb W(\mc G))$ is not strongly connected. Therefore, assume that $(\tilde{\mb A}, \tilde{\mb C}_i)$ is structurally observable $\forall i=1,\ldots,m$ and $\mathcal D(\mb W(\mc G))$ is not strongly connected. The proof follows along the same lines as the proof of necessity in Lemma~\ref{le:iffstructObs} since condition (i) of Theorem~\ref{thm:decstructobs} fails if $\mc D(\mb W(\mc G))$ is not strongly connected. Now, assume that $(\tilde{\mb A}, \tilde{\mb C}_i)$ is structurally observable $\forall i=1,\ldots,m$ and $\mc D(\mb W(\mc G\setminus \mc P_i))$ is not spanned by a disjoint union of cycles. Using the second part of the proof of Lemma~\ref{le:iffstructObs} for $\bar{\mb M}$ corresponding to each of the strongly connected components of $\mc D(\mb W(\mc G\setminus \mc P_i))$ proves the contradiction to condition (ii) of Theorem~\ref{thm:decstructobs} if $\mc D(\mb W(\mc G\setminus \mc P_i))$ is not spanned by a disjoint union of cycles.

\emph{(Sufficiency)}~Assume $\mc D(\mb W(\mc G\setminus \mathcal P_i))$ is spanned by a disjoint union of cycles for all $i=1,\ldots,m$, and $\mc D(\mb W(\mc G))$ is strongly connected. In order to prove sufficiency, we follow similar steps to those in the proof in \cite{AlexandruCDC16} 
and show that the same conditions are satisfied by a more general $\mb W(\mc G)$.

The necessity and sufficiency of condition~(i) and sufficiency of condition~(ii) in Theorem~\ref{thm:decstructobs} follow as in \cite{AlexandruCDC16}, since $\mc D(\mb W(\mc G))$ is assumed to be strongly connected.
The original system \eqref{dynamics}-\eqref{output} is also structurally observable, and, from Theorem~\ref{thm:structObs}, we know there exists a maximum matching $M$ associated to $\mathcal B(\bar{\mathbf A})$ such that, for every \mbox{left-unmatched} vertex~$x\in\mathcal U_L(M)$, there is a distinct sensor associated to it. Expanding the maximum matching $M$ to the augmented system's bipartite state graph $\mathcal B(\tilde{\mathbf A})$, we can match all the vertices $x\in\mathcal U_L(M)$ with a distinct sensor measuring it. This yields that the only possible \mbox{left-unmatched} vertices in $\mathcal{B}(\tilde{\mathbf A})$ are $z\in \mathcal Z$, hence, $\mathcal U_L\cap \mathcal X = \emptyset$. All sensors $z_j$ that are not right-matched by a path from a previously unmatched state vertex are spanned by disjoint cycles, by the assumption on $\mc D(\mb W(\mc G))$. Then, either these \mbox{left-unmatched} vertices are already \mbox{in-neighbors} of sensor $z_i$, or, following a similar procedure as in the proof of Lemma~\ref{le:iffstructObs}, we can find another maximum matching such that $\mathcal U_L\subset \mathcal N_i^-$.



Next, to show that there exists a realization of $\mb W(\mc G)$ that ensures observability of $(\tilde{\mb{A}},\tilde{\mb{C}}_i)$, we leverage the proof of Theorem~\ref{thm:almostAll} in \cite{AlexandruCDC16}. Specifically, we invoke the Popov-Belevitch-Hautus criterion to assess the observability of the system $(\tilde{\mb A},\tilde{\mb C}_i)$. As a result, $\mb W(\mc G)$ must be such that the following equalities hold for $\lambda \in\mathbb{C}, \ i=1,\ldots,m$:\vspace{-0.1cm}
\begin{align*}
\text{rank}\underbrace{\left [ \begin{matrix}
\tilde{\mb A}-\lambda \mb{I}_{n+m}\\
\tilde{\mb C}_i
\end{matrix}\right ]}_{\mb M}& = \text{rank}\left[
\begin{smallmatrix}
\mathbf A-\lambda \mb{I}_{n} &\mathbf{0}_{n\times m}\\
\mathbf C& \mathbf W(\mathcal G)-\lambda \mb{I}_{m}\\
-\ \mathbf c_i^\intercal\ - & \mathbf{0}_{1\times m}\\
\mathbf{0}_{|\mathbb N_i^-|\times n} &\mb{I}_{m}^{\mathbb N_i^-}\end{smallmatrix}
\right]=n+m.
\end{align*}\vspace{-0.3cm}

The structure of $\mb W(\mc G)$ does not allow arbitrary placing of the eigenvalues, as opposed to the proof of Theorem~\ref{thm:almostAll} in~\cite{AlexandruCDC16}. However, we are able to prove that the eigenvalues of $\mb W(\mc G)$ that we cannot place do not affect the rank of $\mb M$. The eigenvalues associated with the cycles in $\mb W(\mc G\setminus \mc P_i)$ can be arbitrarily placed: for each cycle, composed of edges with weights $w^1,\ldots, w^{r_j}$, where $1\leq j \leq |\mc C|$ and $|\mc C|$ is the number of cycles, the associated eigenvalues will have the values: $\lambda_k^j = \sqrt[r_j]{w^1\cdot\ldots\cdot w^{r_j}}e^{2\pi i \frac{k-1}{r_j}},k=1,\ldots,r_j$. The eigenvalues that cannot be placed are associated to the paths $\mc P_i$ and will be zero \cite{Coates59, Reinschke:1988}. The same analysis holds for $\mb A$, i.e., the zero eigenvalues are associated with the paths that are not spanned by cycles. More specifically, the vertices on these paths are exactly the \mbox{left-unmatched} vertices with respect to a maximum matching in $\mc B(\tilde{\mb A},\tilde{\mb C}_i)$. In the \mbox{state-output} digraph $\mc D(\tilde{\mb A}, \tilde{\mb C}_i)$, in order to match these \mbox{left-unmatched} vertices, the paths are extended through the links described by $\mb C$ to the sensors' states, in the communication graph $\mc D(\mb W(\mc G))$. Let $p$ be the number of vertices in the minimum length paths in $\mc D(\bar{\mb A})$ and $\mc D(\mb W(\mc G))$ that are not spanned by cycles, i.e., corresponding to the zero eigenvalues in $\mb A$ and $\mb W(\mc G)$. 
By suitable permutations, we can separate the blocks in $\mb A$ and $\mb W(\mc G)$ associated to the disjoint cycles, denoted symbolically by $\mb A_{\mc C},\mb W_{\mc C}$ and the blocks associated to the linkings $\mc P_i$, respectively $\mb A_{\mc P},\mb W_{\mc P}$:\vspace{-0.1cm}
\[\begin{array}{ll}
\underbar{M}:=\left[
\begin{matrix}
\mathbf A &\mathbf{0}_{n\times m}\\
\mathbf C& \mathbf W(\mathcal G)\end{matrix}
\right] = \left[
\begin{smallmatrix}
 \mb A_{\mc P} & \mb 0 & \ast & \mb 0\\
\mb 0 & \mb W_{\mc P} & \mb 0 & \ast \\
\ast & \mb 0 & \mb A_{\mc C} & \mb 0\\
\mb 0 & \ast & \mb 0 & \mb W_{\mc C}
\end{smallmatrix}
\right].
\end{array}
\]

The eigenvalues in $\mb W_{\mc P}$ can be chosen to be different than the eigenvalues of $\mb A_{\mc C}$, which are non-zero, and different than zero, hence $\underbar{M}$ has $n+m-p$ non-zero eigenvalues. Moreover, the end-vertices of the linkings $\mc P_i$ are measured by the outputs given by $\mb{I}_{m}^{\mathbb N_i^-}$. Therefore, the pair $(\underbar{M}_{(1:p,1:p)},\mb{I}_{m}^{\mathbb N_i^-})$ is observable and, by Popov's criterion,\vspace{-0.1cm}
\[\text{rank} \left[ \begin{array}{c}
 \underbar{M}_{(1:p,1:p)}-\lambda \mb{I}_p\\
 \mb{I}_m^{\mathbb N_i^-}
\end{array}\right] = p.\]
This completes the proof that rank$(\mb M) = n+m$, i.e., the conditions of Theorem~\ref{thm:iffW} are sufficient.
%
\end{proof}


\begin{remark}\label{re:particularization} Theorem~\ref{thm:iffW} accounts for scenarios where the sensors are memoryless, i.e., they do not retain their previous state to integrate it in the overall dynamics. This extends the results in~\cite{AlexandruCDC16}, revisited in Section~\ref{sec:termPR}. 
Specifically, the case where sensors are not readily memoryless leads to the case where the communication graph is strongly connected and has a subgraph spanned by a disjoint union of cycles, since the access of a sensor to its state and incorporation in the overall dynamics corresponds to a self-loop in the communication graph, which is an elementary cycle. \hfill $\diamond$
\end{remark}

\begin{remark}\label{re:radio}
In the context of limited communication decentralized estimation schemes that employ sensors which cannot discern between the contributions coming from their neighbors due to the technology used, i.e., when $\mathcal J_i=\{i\}$ in~\eqref{outputAugmented}, it follows that $\mb A$ is at most rank $n-1$. More specifically, there will be only one possible \mbox{communication-linking} ending at the $i$'th sensor vertex, implying that the state bipartite graph's maximum matching can have at most one \mbox{left-unmatched} vertex. 
\hfill $\diamond$
\end{remark}

Finally, given the necessary and sufficient conditions for the communication graph provided in Theorem~\ref{thm:iffW}, we aim to formulate the problem of designing a minimum cost communication graph, as stated in \emph{Problem~2}, as an integer programming problem. To this end, we leverage 
problems such as the minimum cost maximum matching problem and minimum cost spanning trees~\cite{nemhauser1988integer}. We also need to formulate the conditions on the communication graph, which require to encode the minimum cardinality linkings~$\mc P_i$.

To better visualize the results, we write the communication graph as $\mathcal G=(\mathcal Z,\mathcal E_{\mc Z,\mc Z})$, where $\mc Z$ represents the set of sensor states, and $\mc E_{\mc Z,\mc Z}$ the set of communication links between~the sensors. Let $\mb \Omega_e\in \mathbb{R}^+_{\geq0}\cup\{\infty\}$ be the communication cost incurred by establishing a link $e=(i,j)\in \mathcal E_{\mc Z,\mc Z}$ from sensor~$j$ to sensor~$i$. If we want to obtain a communication graph dealing with memoryless sensors, then we prescribe~$\mb \Omega_{ii} = \infty$, and obtain a finite cost graph as a feasible solution.

Briefly, the constraints that have to be satisfied by $\mc G$ can be described in the following algorithm where steps are addressed simultaneously: for every sensor $i=1,\ldots,m$,\\
1.~Find the number of \mbox{left-unmatched} vertices in $\mc B(\bar{\mb A})$;\\
2.~Find the minimum cost linking $\mc P_i$;\\
3.~Run the minimum cost maximum matching algorithm on $\mc G\setminus \mc P_i$ and select the edges that compose it; and\\
4.~Add the minimum cost edges such that the graph $\mc G$ is strongly connected.

We can leverage some insights provided by the heuristic algorithm provided in~\cite{AlexandruCDC16} 
to obtain an integer programming problem formulation without explicitly computing $\mc P_i$. For a sensor $i$, add a `virtual output' (i.e., not part of the sensing technology but with the same role under this intermediate step) to each of its \mbox{in-neighbors}, according to $\mb I_m^{\mathbb N^-_i}$. Denote this set of vertices by $\mc S^i$. 
Next, let $\mc B(\tilde{\mb A}, \tilde{\mb C}_i, \mc S^i)=(\mc X\cup\mc Z,\mc X\cup\mc Z\cup \mc S^i,$ $\mc E_{\mc X\cup\mc Z,\mc X\cup\mc Z\cup\mc S^i})$ be the \mbox{state-output} bipartite graph of the system with virtual outputs and $\mc D(\tilde{\mb A}, \tilde{\mb C}_i, \mc S^i)$ the associated \mbox{\mbox{state-output}} digraph. 
We then expand the cost structure as follows: assign weights $\mb \Omega^i$ w.r.t.~sensor~$i$ to all the edges that are not in the communication graph $\mc G$ and $\mb \Omega ^i_e = 0$ for $e \in \mathcal{E}_{\mc X\cup \mc Z,\mc X \cup \mc Z}\cup\mathcal{E}_{\mc Z,\mc S^i}$. This setup ensures that the minimum cost maximum matching algorithm in the former \mbox{state-output} bipartite graph will return a matching that partitions the \mbox{state-output} digraph in \mbox{vertex-disjoint} paths and cycles, while incurring the minimum cost. Specifically, the paths will contain those described by~$\mc P_i$, and the rest of the digraph will be spanned by disjoint cycles. Furthermore, notice that since there are no edges from the communication digraph to the state digraph, there can be no cycle spanning both vertices in $\mc X$ and vertices in $\mc Z$.

To state the integer programming formulation, let us denote by $u_e$ the binary variable associated to the existence of edge $e\in\mc D(\tilde{\mb A}, \tilde{\mb C}_i, \mc S^i),i=1,\ldots,m$. For a set $S\subset \mc V$, the cutset $\delta^-(S)\subset \mc E$ represents a subset of edges with the start vertex in $S$ and the end vertex in $\mc V \setminus S$, for a given set of edges $\mc E$ and set of vertices $\mc V$.
Hence, the constraints are given by the matching problem on the \mbox{state-output} digraph and by the strong connectivity of the communication graph, which is imposed via rooted minimum spanning tree for each vertex. For brevity, denote 
$\mc E^i:= \mc E_{\mc Z,\mc Z\cup\mc S^i}$. Thus, we obtain the following 
formulation of \emph{Problem~2}:
\vspace{-0.1cm}
\begin{align*}
	\min\limits_{u_e}&~ \sum_{e\in \cup_{i=1}^m \mc E_{\mc X\cup\mc Z,\mc X\cup\mc Z\cup\mc S^i }}\mb \Omega_e u_e\\ \vspace{-0.1cm}
	s.t&~ \sum_{e\in \delta^-(v)} u_e \leq 1,~ v\in \mc Z,\\ 
	    &~ \sum_{e\in \delta^-(S)} u_e \geq 1, ~\forall~ \emptyset \subsetneq S\subsetneq \mc Z, v\notin S, ~\forall ~ v\in \mc Z,\\
	    &~u_e = 1,~\forall e \in \mc E_{\mc X,\mc X}\cup\mc E_{\mc X,\mc Z}~\text{and}~u_e = 0~\text{otherwise},\\
	&~ u_e \in\{ 0,1\},~\forall e\in \mc E^i,~i = 1,\ldots,m.
\end{align*}\vspace{-0.55cm}

\noindent where the edges in the minimum cost communication graph~$\mc G^\ast$ can be retrieved from $\mc U:= \{e\in \mc E_{\mc Z,\mc Z}| u_e = 1\}$.

The design problem proposed in~\emph{Problem~2} is NP-hard, since it contains as a particular instance the design problem addressed in~\cite{AlexandruCDC16}. 
Hence, a straightforward greedy algorithm can be implemented by sequentially performing the steps described in the pseudo-algorithm above. Nonetheless, the solution will depend on the initial point since, at each iteration, the previous selected edges in $\mc E$ will be set to one, 
which does not guarantee that the final configuration of $\mc G$ has indeed minimum cost. Consequently, we leverage the integer programming formulation which is also known to be NP-hard in general, but which we can solve by resorting to highly optimized off-the-shelf software toolboxes (e.g., YALMIP) that have been efficiently deployed in practice when dealing with large-scale complex problems~\cite{nemhauser1988integer,lofberg2004yalmip}. 

\begin{figure}
 \begin{centering}
\begin{tikzpicture}[>=stealth',shorten >=1pt,node distance=1cm,on grid,initial/.style    ={}]

  \node[state,scale=0.5,fill]          (x2)                        {\textcolor{white}{$x_2$}};
  \node[state,scale=0.5,fill]          (x1) [left =of x2]    {\textcolor{white}{$x_1$}};
  \node[state,scale=0.5,fill]          (x3) [right =of x2]    {\textcolor{white}{$x_3$}};
  \node[state,scale=0.5,fill]          (x4) [right =of x3]    {\textcolor{white}{$x_4$}};
  \node[state,scale=0.5,fill]          (x5) [right =of x4]    {\textcolor{white}{$x_5$}};
  \node[state,scale=0.5,fill=blue]          (y1) [above =of x1,yshift=0.3cm]    {\textcolor{white}{$y_1$}};
  \node[state,scale=0.5,fill=blue]          (y2) [above =of x2,yshift = 0.3cm]    {\textcolor{white}{$y_2$}};
  \node[state,scale=0.5,fill=blue]          (y3) [above =of x3,yshift=0.3cm]    {\textcolor{white}{$y_3$}};
  \node[state,scale=0.5,fill=blue]          (y4) [above =of x4,yshift = 0.3cm]    {\textcolor{white}{$y_4$}};
  \node[state,scale=0.5,fill=blue]          (y5) [above =of x5,yshift=0.3cm]    {\textcolor{white}{$y_5$}};
\tikzset{mystyle/.style={->,red}} 
\path (x1)     edge [mystyle]    node   {} (y1)
	(x2)     edge [mystyle]    node   {} (y2)
	(x3)     edge [mystyle]    node   {} (y3)
	(x4)     edge [mystyle]    node   {} (y4)	
	(x5)     edge [mystyle]    node   {} (y5);

 \path[->]          (x1)  edge   [bend right=20]   node {} (x2);
  \path[->]          (x3)  edge   [bend right=20]   node {} (x2);
  \path[->]          (x3)  edge   [bend right=20]   node {} (x4);
  \path[->]          (x4)  edge   [bend right=20]   node {} (x3);
  \path[->]          (x4)  edge   [bend right=20]   node {} (x5);
  \path[->]          (x5)  edge   [bend right=20]   node {} (x4);
\end{tikzpicture}
  \caption{Plant with 5 state nodes (black) and 5 sensors deployed (blue). The interconnections between the state nodes are depicted in black and the measurement terminals between state nodes and sensors are depicted in red.}
  \vspace{-0.69cm}
  \label{fig:5sensors}
  \end{centering}
\end{figure}
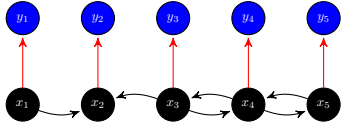

\section{Illustrative example}

Consider the example in Figure~\ref{fig:5sensors} and associate the following cost matrix for the communication links:\vspace{-0.1cm}
{\small \[\mb \Omega(\mc G) = \left[ \begin{smallmatrix}\infty & 1 & 1 & 1 & 1\\ 1 & \infty & 1 & 1 & 1\\ 1 & 1& \infty & 1 & 1 \\ 1 & 1 & 1 & \infty & 1 \\ 1 & 1 & 1 & 1 & \infty \end{smallmatrix}\right].\]}
\vspace{-0.2cm}

The minimum topology for the communication graph~$\mc G$ such that decentralized observability from all sensors is ensured is depicted in Figure~\ref{fig:5sensorsMin}. Since there are two left-unmatched vertices in $\mc B(\bar{\mb A})$, each sensor should have at least two in-neighbors. Here, we illustrate how decentralized observability from sensor 1 is achieved. 
We need to find a linking $\mc P_1$ such that $\mc G\setminus \mc P_1$ is spanned by a disjoint union of cycles. In this case, pick a maximum matching in $\mc B(\bar{\mb A},\bar{\mb C})$ composed of the edges $(x_1,y_1),(x_2,y_2),(x_3,x_2),(x_4,x_3),(x_5,x_4)$. The linking $\mc P_1$ from the matched sensors to the in-neighbors of sensor 1 can be chosen as $(y_1,y_4,y_5)\cup(y_2,y_3)$. We obtain that $\mc G\setminus \mc P_1=\emptyset$, which is trivially spanned by a union of disjoint cycles. 

\begin{figure}[b]
 \begin{centering}
   \vspace{-0.69cm}
\begin{tikzpicture}[>=stealth',shorten >=1pt,node distance=1cm,on grid,initial/.style    ={}]
  \node[state,scale=0.5,fill]          (x2)                        {\textcolor{white}{$x_2$}};
  \node[state,scale=0.5,fill]          (x1) [left =of x2]    {\textcolor{white}{$x_1$}};
  \node[state,scale=0.5,fill]          (x3) [right =of x2]    {\textcolor{white}{$x_3$}};
  \node[state,scale=0.5,fill]          (x4) [right =of x3]    {\textcolor{white}{$x_4$}};
  \node[state,scale=0.5,fill]          (x5) [right =of x4]    {\textcolor{white}{$x_5$}};
  \node[state,scale=0.5,fill=blue]          (y1) [above =of x1,yshift=0.3cm]    {\textcolor{white}{$y_1$}};
  \node[state,scale=0.5,fill=blue]          (y2) [above =of x2,yshift = 0.3cm]    {\textcolor{white}{$y_2$}};
  \node[state,scale=0.5,fill=blue]          (y3) [above =of x3,yshift=0.3cm]    {\textcolor{white}{$y_3$}};
  \node[state,scale=0.5,fill=blue]          (y4) [above =of x4,yshift = 0.3cm]    {\textcolor{white}{$y_4$}};
  \node[state,scale=0.5,fill=blue]          (y5) [above =of x5,yshift=0.3cm]    {\textcolor{white}{$y_5$}};
\tikzset{mystyle/.style={->,red}} 
\path (x1)     edge [mystyle]    node   {} (y1)
	(x2)     edge [mystyle]    node   {} (y2)
	(x3)     edge [mystyle]    node   {} (y3)
	(x4)     edge [mystyle]    node   {} (y4)	
	(x5)     edge [mystyle]    node   {} (y5);

 \path[->]          (x1)  edge   [bend right=20]   node {} (x2);
  \path[->]          (x3)  edge   [bend right=20]   node {} (x2);
  \path[->]          (x3)  edge   [bend right=20]   node {} (x4);
  \path[->]          (x4)  edge   [bend right=20]   node {} (x3);
  \path[->]          (x4)  edge   [bend right=20]   node {} (x5);
  \path[->]          (x5)  edge   [bend right=20]   node {} (x4);

  \path[->,blue]          (y1)  edge   [bend right=30]   node {} (y3);
  \path[->,blue]          (y1)  edge   [bend right=30]   node {} (y4);
  \path[->,blue]          (y2)  edge   [bend right=20]   node {} (y3);
  \path[->,blue]          (y3)  edge   [bend right=20]   node {} (y1);
  \path[->,blue]          (y3)  edge   [bend right=20]   node {} (y2);
  \path[->,blue]          (y3)  edge   [bend right=20]   node {} (y4);
  \path[->,blue]          (y3)  edge   [bend right=40]   node {} (y5);
  \path[->,blue]          (y4)  edge   [bend right=20]   node {} (y5); 
  \path[->,blue]          (y5)  edge   [bend right=30]   node {} (y1);
  \path[->,blue]          (y5)  edge   [bend right=30]   node {} (y2);
  \path[->,blue]          (y5)  edge   [bend right=20]   node {} (y4);
\end{tikzpicture}
  \caption{Minimum communication topology that ensure decentralized observability for the plant in Figure~\ref{fig:5sensors}. }
  \label{fig:5sensorsMin}
  \end{centering}
\end{figure}
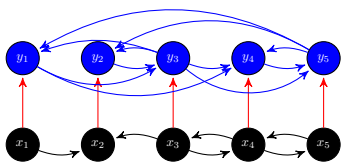
\vspace{-0.1cm}
\section{Conclusions}

In this paper, we have extended the limited communication decentralized estimation schemes to cope with general scenarios and provided necessary and sufficient conditions for the communication graph such that retrieval of the state of the system and sensors is possible. In particular, the present extension enables the deployment of limited communication decentralized estimation schemes in the scenarios where the sensors are memoryless, 
and where the sensors do not have the capacity to discern the contributions and/or state of other sensors. Furthermore, we cast the design problem under communication costs, i.e., the problem of determining the minimum cost communication graph required to implement a limited communication decentralized estimation scheme, as an integer programming problem. 
This formulation enables the use of off-the-shelf software toolboxes that are reliable in practice when dealing with large-scale complex problems.

Future research will focus on proposing energy-efficient communication protocols between the sensors when subject to constrained energy budgets. Towards this goal, it is important to understand the trade-offs between communication and the information contained in the sensors states. In particular, we aim to quantify and classify the role of the sensors' state dimension in the estimation process and accuracy, which can be key when adding communication infrastructure is prohibitive. Our findings suggest that as the number of dimensions of the exchanged states increases, fewer communication links are required to guarantee decentralized observability. Moreover, one can design the dimension of the sensors's memory such that no additional links are necessary. 



\footnotesize

\bibliographystyle{IEEEtran}
\bibliography{IEEEabrv,acc2016_2}

\end{document}